\documentclass[a4paper]{article}
\usepackage{graphicx}
\usepackage{amsfonts}
\usepackage{amssymb}
\usepackage{amsmath,latexsym,amsbsy}
\usepackage{amsthm}
\usepackage[square,comma,numbers,sort&compress]{natbib}

\pdfoutput=1

\newtheorem{theorem}{Theorem}
\newtheorem{lemma}{Lemma}
\newtheorem{remark}{Remark}

\def\R{\mathbb R}
\def\pa{{\partial\Omega}}

\def\h{h}
\def\dc{\partial}
\def\ee{\zeta}

\title{On the exponential decay of Laplacian eigenfunctions in planar domains with branches}

\author{Binh T. Nguyen$^1$, Andrey L. Delytsin$^2$ and Denis S. Grebenkov$^1$}

\begin{document}

\maketitle

\begin{center}
$^1$ Laboratoire de Physique de la Mati\`ere Condens\'ee \\
CNRS -- Ecole Polytechnique, F-91128 Palaiseau, France \\
$^2$ Mathematical Department of the Faculty of Physics, \\
Moscow State University, 119991 Moscow, Russia  \\
~ \\
Corresponding author: denis.grebenkov@polytechnique.edu
\end{center}

\date{Received: date / Accepted: date}

\maketitle

\begin{abstract}
We consider the eigenvalue problem for the Laplace operator in a
planar domain which can be decomposed into a bounded domain of
arbitrary shape and elongated ``branches'' of variable cross-sectional
profiles.  When the eigenvalue is smaller than a prescribed threshold,
the corresponding eigenfunction decays exponentially along each
branch.  We prove this behavior for Robin boundary condition and
illustrate some related results by numerically computed
eigenfunctions.
\end{abstract}

\section{Introduction}

The geometrical structure of Laplacian eigenfunctions has been
thoroughly investigated (see the review \cite{Grebenkov13} and
references therein).  When a domain can be seen as a union of two (or
many) subdomains with narrow connections, some low-frequency
eigenfunctions can be found localized (or trapped) in one subdomain
and of small amplitude in other subdomains.  Qualitatively, an
eigenfunction cannot ``squeeze'' through a narrow connection when its
typical wavelength is larger than the connection width.  This
qualitative picture has found many rigorous formulations for dumbbell
shapes and classical and quantum waveguides
\cite{Ashbaugh90,Goldstone92,Duclos95,Exner04,Linton07,Jimbo09,Filoche09,Odenski10,Delitsyn12a,Delitsyn12b,Filoche12}.
Numerical and experimental evidence for localization in
irregularly-shaped domains was also reported
\cite{Sapoval89,Sapoval91,Sapoval93,Sapoval97,Haeberle98,Even99,Hebert99,Felix07,Heilman10}.

In a recent paper, we considered the Laplacian eigenvalue problem
\begin{equation}
\label{eq:Laplace_in_D}
\begin{split}
\Delta u + \lambda u & = 0 \quad \text{in}~ \Omega, \\
 u & = 0 \quad \text{on}~ \partial \Omega,\\
\end{split}
\end{equation} 
for a large class of domains $\Omega$ in $\R^d$ ($d = 2,3,...$) which
can be decomposed in a ``basic'' bounded domain $V$ and a branch $Q$
of a variable cross-sectional profile \cite{Delitsyn12a}.  We proved
that if the eigenvalue $\lambda$ is smaller than the smallest
eigenvalue $\mu$ among all cross-sections of the branch, then the
associated eigenfunction $u$ exponentially decays along that branch:

\begin{theorem}
\label{theo:Dirichlet}
Let $\Omega\subset \R^d$ ($d = 2,3,...$) be a bounded domain with a
piecewise smooth boundary $\pa$ and let $Q(z) = \Omega\cap \{
x\in\R^d~:~ x_1 = z\}$ the cross-section of $\Omega$ at $x_1 = z \in
\R$ by a hyperplane perpendicular to the coordinate axis $x_1$
(Fig. \ref{fig:branches}).  Let
\begin{equation*}
z_1 = \inf\{ z\in \R~:~ Q(z) \ne \emptyset\} , \qquad z_2 = \sup\{ z\in \R~:~ Q(z) \ne \emptyset\} ,
\end{equation*}
and we fix some $z_0$ such that $z_1 < z_0 < z_2$.  Let $\mu(z)$ be
the first eigenvalue of the Laplace operator in $Q(z)$, with Dirichlet
boundary condition on $\partial Q(z)$, and $\mu = \inf\limits_{z\in
(z_0,z_2)} \mu(z)$.  Let $u$ be a Dirichlet-Laplacian eigenfunction in
$\Omega$ satisfying (\ref{eq:Laplace_in_D}), and $\lambda$ the
associate eigenvalue.  If $\lambda < \mu$, then
\begin{equation}
\label{eq:expon_estimate}
\|u\|_{L_2(Q(z))} \leq \|u\|_{L_2(Q(z_0))} \exp(-\beta \sqrt{\mu-\lambda}~(z - z_0))   \quad  (z \geq z_0),
\end{equation} 
with $\beta = 1/\sqrt{2}$.  Moreover, if $(e_1 \cdot n(x)) \geq 0$ for
all $x\in\pa$ with $x_1 > z_0$, where $e_1$ is the unit vector
$(1,0,...,0)$ in the direction $x_1$, and $n(x)$ is the normal vector
at $x\in\pa$ directed outwards the domain, then the above inequality
holds with $\beta = 1$.
\end{theorem}
In this theorem, a domain $\Omega$ is arbitrarily split into two
subdomains, a ``basic'' domain $\Omega_1$ (with $x_1 < z_0$) and a
``branch'' $\Omega_2$ (with $x_1 > z_0$), by the hyperplane at $x_1 =
z_0$ (the coordinate axis $x_1$ can be replaced by any straight line).
Under the condition $\lambda < \mu$, the eigenfunction $u$
exponentially decays along the branch $\Omega_2$.  Note that the
choice of the splitting hyperplane (i.e., $z_0$) determines the
threshold $\mu$.  Since $\mu$ is independent of the basic domain $V$,
one can impose any boundary condition on $\partial \Omega_1$ (that
still ensures the self-adjointness of the Laplace operator).  In turn,
the Dirichlet boundary condition on the boundary of the branch
$\Omega_2$ was relevant.  Many numerical illustrations for this
theorem were given in \cite{Delitsyn12a}.

\begin{figure}
\begin{center}
\includegraphics[width=120mm]{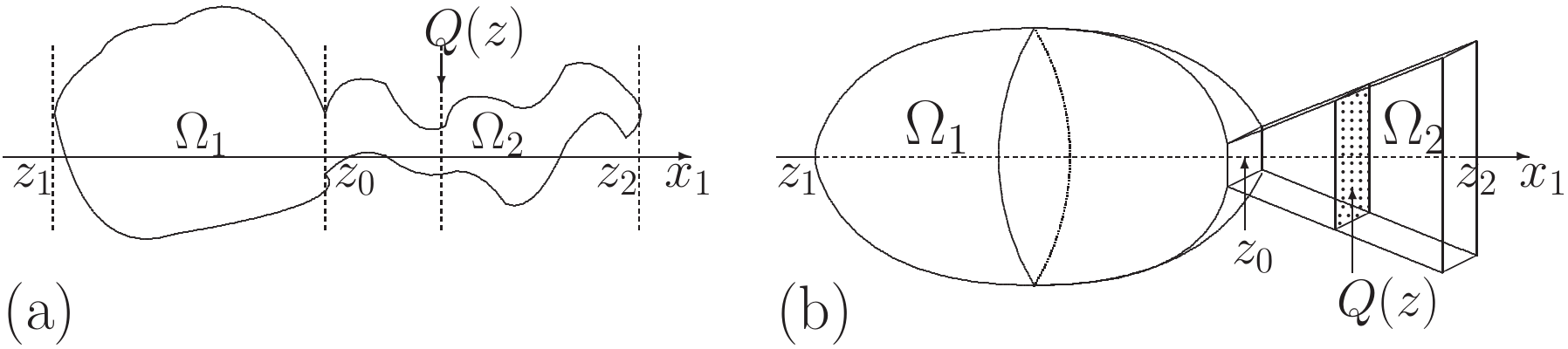}
\end{center}
\caption{
Two examples of a bounded domain $\Omega = \Omega_1 \cup \Omega_2$
with a branch $\Omega_2$ of variable cross-sectional profile $Q(z)$.
When the eigenvalue $\lambda$ is smaller than the threshold $\mu$, the
associated eigenfunction exponentially decays in the branch $\Omega_2$
and is thus mainly localized in $\Omega_1$.  Note that the branch
itself may even be increasing.}
\label{fig:branches}
\end{figure}

\begin{remark}
It is worth stressing that the sufficient condition $\lambda < \mu$
involves purely spectral information: the eigenvalue $\lambda$ in the
whole domain and the smallest eigenvalue $\mu$ over all
cross-sections.  A simple geometrical condition on the basic domain
$\Omega_1$ can be formulated through the inradius $\rho$ of $\Omega_1$
(or $\Omega$), i.e., the radius of the largest inscribed ball
$B_\rho$.  Since $B_\rho \subset \Omega$, the first Dirichlet
eigenvalue $\lambda$ is bounded as $\lambda \leq \lambda_1(B_\rho) =
j_{\frac{d}{2}-1}^2/\rho^2$, where $j_{\frac{d}{2}-1}$ is the first
positive zero of the Bessel function $J_{\frac{d}{2}-1}(z)$.  A
sufficient geometrical condition for getting exponentially decaying
eigenfunction is then
\begin{equation}
\label{eq:rho}
\rho > j_{\frac{d}{2}-1}/\sqrt{\mu} .
\end{equation}
For planar domains, the inequality yields $\rho/b > j_0/\pi$ where $b$
is the largest width of the branch, and $j_0 \simeq 2.4048$.  This
inequality includes only the inradius of $\Omega_1$ (or $\Omega$) and
the largest width of the branch, while the length of the branch can be
varied arbitrarily.  For instance, the localization in the basic
domain $\Omega_1$ may hold even when the area of the branch $\Omega_2$
is arbitrarily large, as compared to the area of $\Omega_1$.
\end{remark}

\begin{remark}
For higher dimensions ($d \geq 3$), the localization may sound even
more striking, as the ``branch'' has to be ``narrow'' only in one
direction (Fig. \ref{fig:branches}).  For instance, if the branch
$\Omega_2$ has a constant width $b$ in one direction, then the
smallest eigenvalue $\mu$ in its cross-sections is greater than
$\pi^2/b^2$.  If the inradius $\rho$ of $\Omega_1$ is greater than $b
j_{\frac{d}{2}-1}/\pi$ then the inequality (\ref{eq:rho}) holds, and
at least the first eigenfunction is localized in $\Omega_1$.  In the
three-dimensional space, $j_{\frac{1}{2}} = \pi$ so that the inradius
has to be just greater than $b$: $\rho > b$.  A simple example is a
domain decomposed into the unit cube $\Omega_1 = \{ (x,y,z)\in\R^3~:~
0<x<1,~0<y<1,~0<z<1\}$ and a parallelepiped $\Omega_2 = \{
(x,y,z)\in\R^3~:~ 1<x<L_x,~ -L_y<y<L_y,~ 0<z<b\}$.  When $b < 1/2$,
the first eigenfunction is localized in the cube $\Omega_1$, whatever
the lateral spatial sizes $L_x$ and $L_y$ of the ``branch'' are.
\end{remark}

In the remainder of the paper, we extend the above result to the
Laplace operator in planar domains with Robin boundary condition.  We
also provide several numerical illustrations of localized
eigenfunctions in planar domains in Sec. \ref{sec:illustrations}.

\section{Extension for Robin boundary condition}

We consider the eigenvalue problem for the Laplace operator in a
planar domain $\Omega$ with Robin boundary condition on a piecewise
smooth boundary $\pa$:
\begin{equation}
\label{eq:Laplace_R}
\begin{split}
  \Delta u + \lambda u & = 0 \quad \text{in}~ \Omega ,  \\
  \frac{\dc u}{\dc n} + \h u & = 0 \quad \text{on}~ \pa , \\
\end{split}
\end{equation}
where $h$ is a nonnegative function, and $\partial/\partial n$ is the
normal derivative directed outwards the domain.  In that follows, we
prove the following
\begin{theorem}
\label{theo:Robin}
Let $\Omega = \Omega_1 \cup \Omega_2$, where $\Omega_1\subset \R^2$ is
a bounded domain, and
\begin{equation}
\Omega_2 = \left\{ (x,y)\in\R^2~:~ 0 < x < a, ~ y_1(x) < y < y_2(x) \right\}, 
\end{equation}
is a branch of length $a > 0$ and of variable cross-sectional profile
which is defined by two functions $y_1,y_2\in C^1([0,a])$ such that
$y_2(a) = y_1(a)$, $y'_1(x) \geq 0$ and $y'_2(x) \leq 0$.  Let $u$ and
$\lambda$ be an eigenfunction and eigenvalue of $\Omega$ satisfying
the eigenvalue problem (\ref{eq:Laplace_R}), with a nonnegative
function $h$.  We define $\mu = \inf\limits_{x_0 < x < a}
\mu_1(x)$ where $\mu_1(x)$ is the first eigenvalue of the
Laplace operator in the cross-sectional interval $[y_1(x),y_2(x)]$:
\begin{equation}
\label{eq:eigen_cross}
\begin{split}
 v''(y) + \mu_1(x) v(y) &= 0 \quad (y_1(x) < y < y_2(x))\\
    v'(y) - h_1(x) v(y) &= 0 \quad (y = y_1(x))  \\
    v'(y) + h_2(x) v(y) &= 0 \quad (y = y_2(x))  \\
\end{split}
\end{equation}
where 
\begin{equation}
\label{eq:hi}
h_i(x) \equiv \h(y_i(x)) \sqrt{1 + [y'_i(x)]^2} \qquad (i=1,2).
\end{equation}

If $\lambda < \mu$, then
\begin{equation}
\label{eq:expon_estimateR}
\|u\|_{L_2(\Omega(x))} \leq \|u\|_{L_2(\Omega(x_0))} \exp(-\beta \sqrt{\mu-\lambda}~(x - x_0))   \quad  (x \geq x_0) ,
\end{equation} 
where $\beta = 1/\sqrt{2}$ and $\Omega(x_0) = \{ (x,y)\in\R^2~:~ x_0 <
x < a,~ y_1(x) < y < y_2(x) \}$.
\end{theorem}
\begin{proof}
The proof relies on Maslov's differential inequality and follows the
scheme that we used in \cite{Delitsyn12a} for Dirichlet boundary
condition.  We consider the squared $L_2$-norm of the eigenfunction
$u$ in the ``subbranch'' $\Omega(x_0)$:
\begin{equation*}
I(x_0) = \int \limits_{\Omega(x_0)} dx dy~ u^2  = \int\limits_{x_0}^a dx \int\limits_{y_1(x)}^{y_2(x)} dy ~ u^2(x,y)
\end{equation*}
and derive the inequality for its second derivative: 
\begin{equation}
\label{eq:Ipp}
I''(x_0) \geq 2(\mu - \lambda)I_0(x_0).
\end{equation}

(i) From the first derivative
\begin{equation*}
I'(x_0) = - \int \limits_{y_1(x_0)}^{y_2(x_0)} dy~ u^2(x_0, y)  ,
\end{equation*}
we obtain
\begin{equation*}
I''(x_0) = - 2 \int \limits_{y_1(x_0)}^{y_2(x_0)} dy ~ u \frac{\dc u}{\dc x} - y'_2(x_0) u^2(x_0, y(x_0)) + y'_1(x_0) u^2(x_0, y_1(x_0)) 
\geq - 2 \int \limits_{y_1(x_0)}^{y_2(x_0)} dy ~ u \frac{\dc u}{\dc x} ,
\end{equation*}
where we used the conditions $y'_2(x) \leq 0$ and $y'_1(x) \geq 0$.
Taking into account that
\begin{eqnarray*}
 - \int \limits_{y_1(x_0)}^{y_2(x_0)} dy ~ u \frac{\dc u}{\dc x} &=&
- \int \limits_{y_1(x_0)}^{y_2(x_0)} dy~ u \frac{\dc u}{\dc x} + \int \limits_{S(x_0)} dS~ u \frac{\dc u}{\dc n} + \int \limits_{S(x_0)} dS~ \h u^2  \\
&=& \int \limits_{\Omega(x_0)} dx dy~ {\rm div} (u \nabla u) + \int \limits_{S(x_0)} dS~ \h u^2   \\
&=& \int \limits_{\Omega(x_0)} dx dy~ (\nabla u, \nabla u) + \int \limits_{\Omega(x_0)} dx dy ~ u \Delta u  + \int \limits_{S(x_0)} dS~ \h  u^2  \\
&=& \int \limits_{\Omega(x_0)} dx dy~ (\nabla u, \nabla u) - \lambda \int \limits_{\Omega(x_0)} dx dy ~ u^2 + \int \limits_{S(x_0)} dS~ \h u^2  
\end{eqnarray*}
where 
\begin{equation*}
S = S_1 \cup S_2, \qquad  S_i = \{ (x,y)\in \R^2~:~ 0<x<a, ~ y = y_i(x)\} \quad (i=1,2)
\end{equation*}
is the ``lateral'' boundary of $\Omega_2$, we obtain
\begin{equation*}
\begin{split}
I''(x_0) &\geq 2 \int \limits_{S(x_0)} dS~ \h  u^2 + 2 \int \limits_{\Omega(x_0)} dx dy \biggl(\frac{\dc u}{\dc y}\biggr)^2 
- 2 \lambda \int\limits_{\Omega(x_0)} dx dy ~ u^2 \\
& = 2 \int \limits_{x_0}^a dx \biggl\{ \h_2(x)~u^2(x,y_2(x)) + \h_1(x) ~ u^2(x,y_1(x))  
 + \int \limits_{y_1(x)}^{y_2(x)} dy \biggl[\biggl(\frac{\dc u}{\dc y}\biggr)^2 - \lambda u^2\biggr]\biggr\} , \\ 
\end{split}
\end{equation*}
where $h_i(x)$ is defined by Eq. (\ref{eq:hi}).

According to the Rayleigh principle, the first eigenvalue $\mu_1(x)$
of the eigenvalue problem (\ref{eq:eigen_cross}) on the
cross-sectional interval $[y_1(x), y_2(x)]$ can be written as
\begin{equation}
\label{eq:Rayleigh}
\mu_1(x) = \inf \limits_{u \in H^1([y_1(x), y_2(x)])} \frac{ 
\biggl[\h_2(x) ~ u^2(x,y_2(x)) + \h_1(x) ~ u^2(x,y_1(x)) \biggr]
+ \int \limits_{y_1(x)}^{y_2(x)} dy (\frac{\dc u}{\dc y})^2}{\int \limits_{y_1(x)}^{y_2(x)} dy ~u^2(x,y) } ,
\end{equation}
from which we get
\begin{equation*}
\begin{split}
I''(x_0) & \geq 2\int \limits_{x_0}^a dx (\mu_1(x) - \lambda) \int \limits_{y_1(x)}^{y_2(x)} dy ~u^2(x,y)
\geq 2(\mu - \lambda) \int \limits_{x_0}^a dx  \int \limits_{y_1(x)}^{y_2(x)} dy~ u^2(x,y) = 2(\mu - \lambda) I(x_0) , \\
\end{split}
\end{equation*}
where $\mu = \inf\limits_{x_0 < x < a} \mu_1(x)$.  That completes the
first step.  

(ii) We easily check the following relations:
\begin{equation}
\label{eq:Ipp2}
I(a) = 0, \qquad I'(a) = 0, \qquad I(x_0) \ne 0 \quad (0<x_0<a), \qquad I'(x_0) < 0 \quad (0<x_0<a) .
\end{equation}
Note that the second relation relies on the assumption that $y_1(a) =
y_2(a)$.

(iii) From the inequality (\ref{eq:Ipp}) and relations
(\ref{eq:Ipp2}), an elementary derivation implies the inequality
(\ref{eq:expon_estimateR}).  In fact, one multiplies (\ref{eq:Ipp}) by
$I'(x_0)$, integrates from $x_0$ to $a$, takes the square root and
divides by $I(x_0)$ and integrates again from $x_0$ to $x$ (see
\cite{Delitsyn12a} for details).
\end{proof}

The statement of Theorem \ref{theo:Robin} for Robin boundary condition
is weaker than that of Theorem \ref{theo:Dirichlet} in several
aspects:
\begin{itemize}
\item
Theorem \ref{theo:Robin} employes an explicit parameterization of the
branch through smooth height functions $y_1$ and $y_2$; in particular,
the statement is limited to planar domains.

\item
The branch has to be non-increasing (conditions $y_1(x)\geq 0$ and
$y_2(x)\leq 0$) and vanishing at the end (condition $y_1(a) =
y_2(a)$).

\item
The inequality (\ref{eq:expon_estimateR}) characterizes the $L_2$-norm
of the eigenfunction in the distant part of the branch, $\Omega(x)$,
while the inequality (\ref{eq:expon_estimate}) provided an estimate
at the cross-section $Q(x)$.

\item
The decay rate in Eq. (\ref{eq:expon_estimateR}) involves the
coefficient $\beta = 1/\sqrt{2}$ while the inequality
(\ref{eq:expon_estimate}) for non-increasing branches was proved for
$\beta = 1$.

\end{itemize}
These remarks suggest that the statement of theorem \ref{theo:Robin}
can be further extended while certain conditions may be relaxed.

We also note that the solution of the eigenvalue problem
(\ref{eq:eigen_cross}) has an explicit form
\begin{equation}
v(y) = c_1 \sin(\alpha y) + c_2 \cos(\alpha y),
\end{equation}
with two constants $c_1,c_2$ and $\mu_1(x) = \alpha^2$, while the
boundary conditions at the endpoints $y = y_1(x)$ and $y = y_2(x)$,
\begin{equation*}
\begin{split}
c_1 [-h_1 \sin(\alpha y_1) + \alpha \cos(\alpha y_1)] + c_2[-h_1 \cos(\alpha y_1) - \alpha \sin(\alpha y_1)] & = 0, \\
c_1 [h_2 \sin(\alpha y_2) + \alpha \cos(\alpha y_2)] + c_2[h_2 \cos(\alpha y_2) - \alpha \sin(\alpha y_2)] & = 0, \\
\end{split}
\end{equation*}
yield a closed equation on $\alpha$:
\begin{equation}
\begin{split}
& (\alpha^2 + h_1 h_2) \sin \alpha(y_2-y_1) + \alpha(h_1+h_2) \cos \alpha (y_2 - y_1) = 0  \\
\end{split}
\end{equation}
(here $h_{1,2}$ and $y_{1,2}$ depend on $x$).  This equation has
infinitely many solutions that can be found numerically.  The first
positive solution will determine $\mu_1(x)$.


\section{Illustrations}
\label{sec:illustrations}

In order to illustrate the geometrical structure of Laplacian
eigenfunctions, we compute them for several simple domains.  For all
considered examples, we impose Dirichlet boundary condition for the
sake of simplicity.

\subsection{Sine-shaped branches}

We consider the planar domain $\Omega$ composed of a basic domain
$\Omega_1$ (square of side $L$) and a branch $\Omega_2$ of constant
profile:
\begin{equation}
\label{eq:definition_Qi}
\Omega = \{(x,y) \in \R^2~:~ x \in (0,a),~ y \in (f(x),f(x)+b)\}.
\end{equation}
For this example, we choose $f(x) = \sin(x)$, fix $b = 1$, $L = 1.54$
and take several values for the length $a$.  Since the inradius of the
square $\Omega_1$ is greater than $j_{0}/\pi$, the first eigenvalue
$\lambda$ in these domains is smaller than $\mu = \pi^2$ for any
length $a$ so that the first eigenfunction should be localized in
$\Omega_1$ and exponentially decay along the branch $\Omega_2$.  This
behavior is illustrated on Fig. \ref{fig:sine_branch}.

\begin{figure}
\centering
\includegraphics[width=120mm]{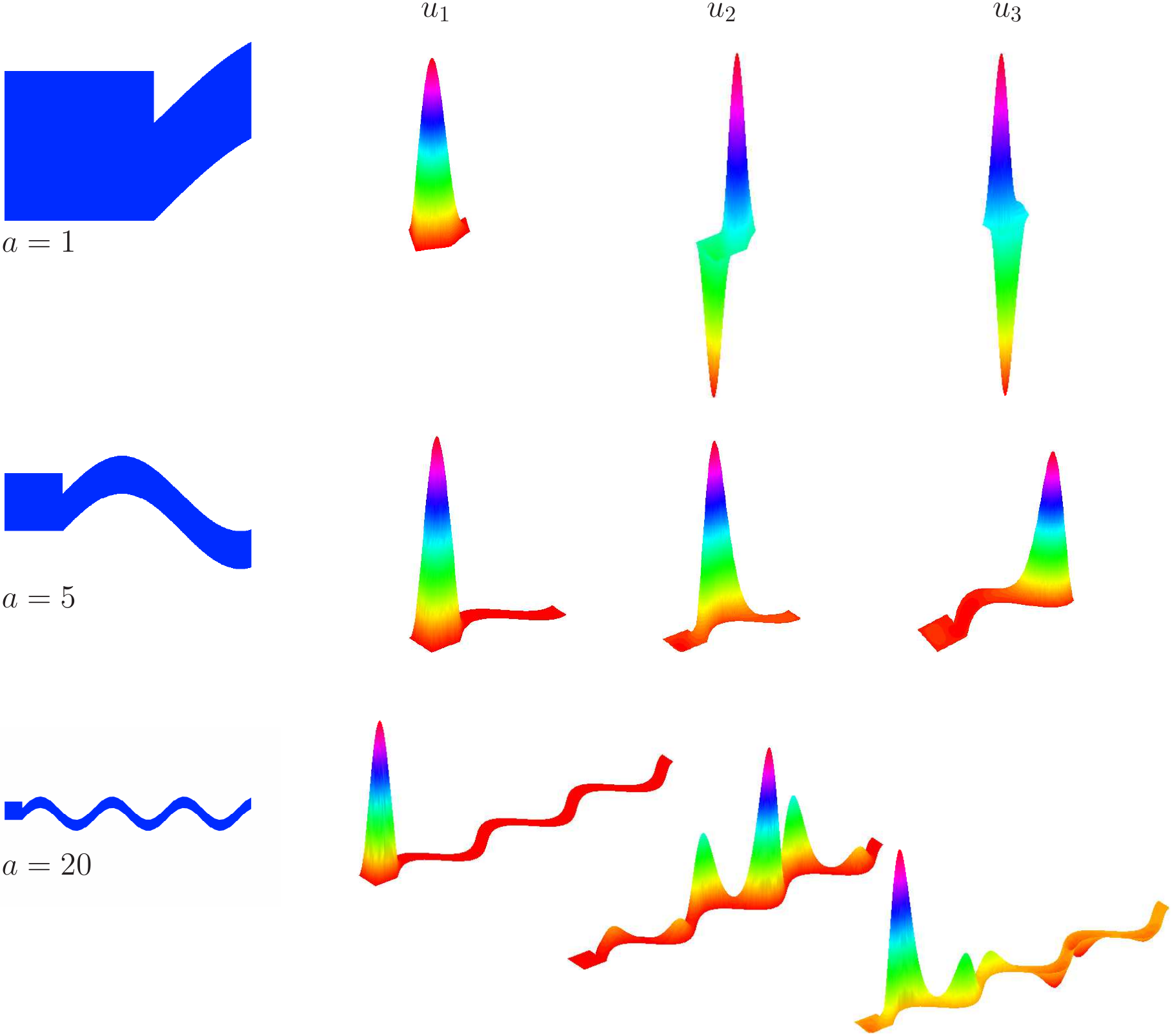}
\caption{ 
First three Dirichlet Laplacian eigenfunctions for the planar domain
$\Omega$ with sine-shaped branches, with $b=1$, $L = 1.54$ and $a=1$,
$a = 5$ and $a = 20$.  The first eigenfunction exponentially decays
along the branch $\Omega_2$, while the second and third eigenfunctions
do not.  The localization occurs in spite of the fact that the area of
$\Omega_1$ presents only $7.15\%$ of the total area for the last
domain with $a = 20$.  Note that the third eigenfunction for $a=5$ is
localized at the end of the branch $\Omega_2$.}
\label{fig:sine_branch}
\end{figure}

\subsection{Star-shaped domains}

Figure \ref{fig:Davies_domain} shows the first five Dirichlet
Laplacian eigenfunctions for a ``star-shaped'' domain which is formed
by a disk with many elongated triangles.  The inradius of this domain
is much greater than the largest width of triangular branches that
implies localization of the first eigenfunction.  One can see that all
the five eigenfunctions are localized in the disk and exponentially
decay along the branches.

\begin{figure}
\centering
\includegraphics[width=120mm]{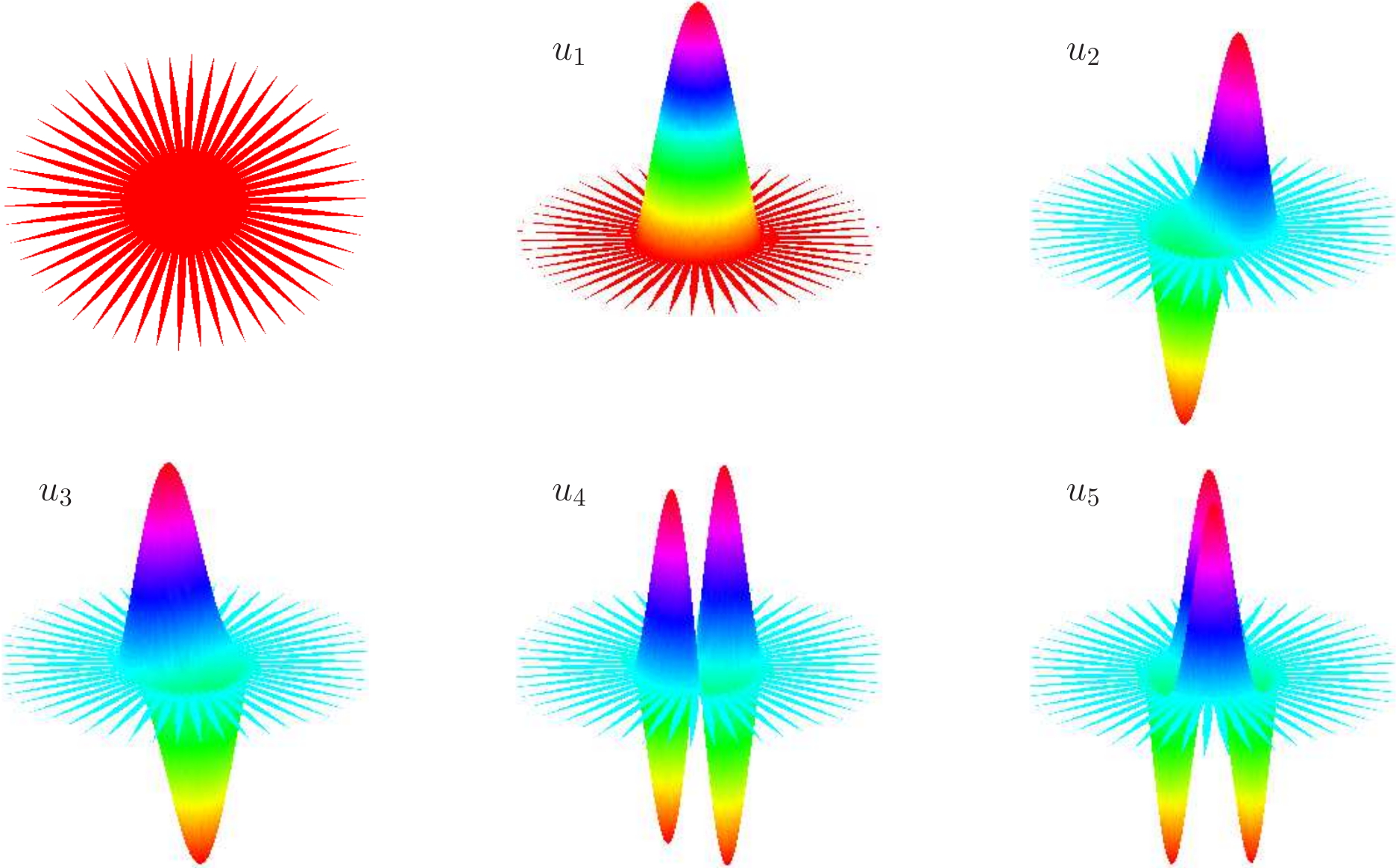}
\caption{
Localization of the first five Dirichlet Laplacian eigenfunctions in a
domain $\Omega$ with $51$ branches.}
\label{fig:Davies_domain}
\end{figure}

\subsection{Elongated polygons}
\label{sec:localization_elongated_triangle}

As we discussed at the beginning, the separation into a basic domain
and a branch is conventional.  We illustrate this point by showing the
exponential decay of the first Dirichlet Laplacian eigenfunction in
elongated polygons for which the ratio between the diameter and the
inradius is large enough.  We start by considering a right triangle
then extend the construction to general elongated polygons.

We consider a rectangle of sides $a$ and $b$ ($a \geq b$) on which a
right triangle $\Omega$ with legs $c$ and $d$ is constructed as shown
on Fig. \ref{fig:triangle}.  Note that the triangle is uniquely
defined by one leg (e.g., $d$), while the other leg is $c = ad/(d-b)$.
The vertical line at $x = a$ splits the triangle $\Omega$ into two
subdomains: $\Omega_1$ (a trapeze) and $\Omega_2$ (a triangle).  For
fixed $a$ and $b$, we are searching for a sufficient condition on $d$
under which the eigenfunction $u$ satisfying the eigenvalue problem
(\ref{eq:Laplace_in_D}), is localized in $\Omega_1$ and exponentially
decays along the subdomain $\Omega_2$.

\begin{figure}
\begin{center}
\includegraphics[width=60mm]{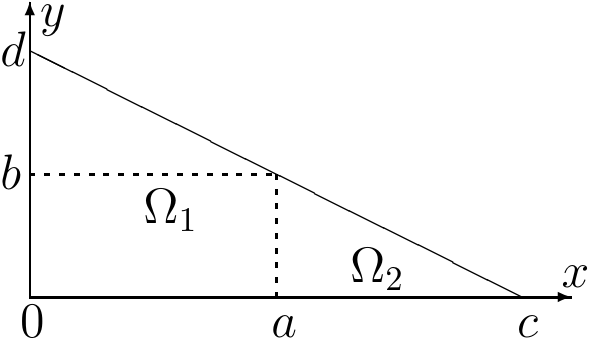}  \hskip 10mm
\includegraphics[width=80mm]{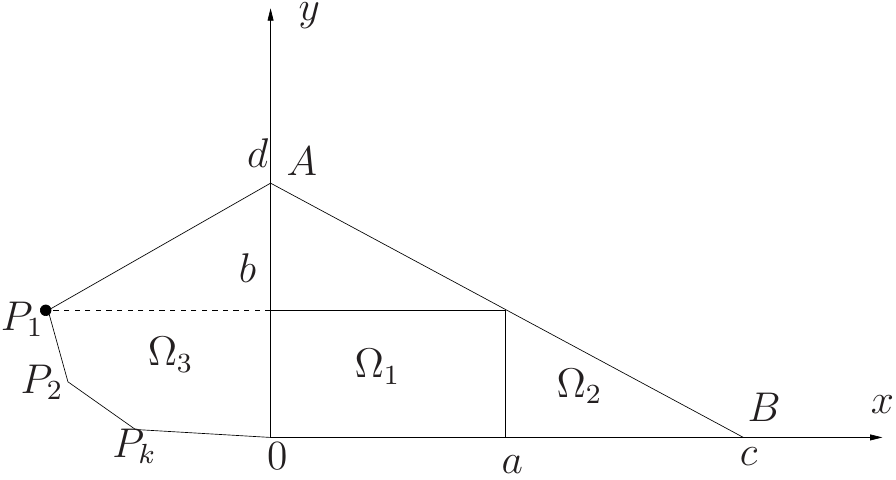}
\end{center}
\caption{
{\bf Left:} A right triangle $\Omega$ is decomposed into a trapeze
$\Omega_1$ and a right triangle $\Omega_2$.  {\bf Right:} An elongated
polygon $P$ with $n$ vertices $\left\{{0, B, A, P_1, P_2, \dots,
P_k}\right\}$, where $k=n-3$. Here, $\Omega_3$ is the polygon
including $n-1$ vertices $\left\{{0, A, P_1, P_2, \dots,
P_k}\right\}$. }
\label{fig:triangle}
\end{figure}

\begin{lemma}
Let $\Omega$ be a right triangle $\Omega$ defined by fixed $a,b,d>0$
($a\geq b$) (Fig. \ref{fig:triangle}).  Let $\ee \equiv (d/b) - 1$,
$P_A(\ee)$ and $P_B(\ee)$ be two explicit polynomials defined by
Eq. (\ref{eq:PA_PB}), and $\ee_0 \approx 0.0131$ is the zero of
$P_A(\ee)$.  If both inequalities
\begin{equation}
\label{eq:ineq}
\ee > \ee_0 , \qquad \frac{a^2}{b^2} > \frac{P_B(\ee)}{P_A(\ee)} 
\end{equation}
are fulfilled, then the first eigenfunction $u$ of the Laplace
operator in $\Omega$ with Dirichlet boundary condition exponentially
decays in $\Omega_2$.
\end{lemma}
\begin{proof}
The proof of the exponential decay relies on Theorem
\ref{theo:Dirichlet}.  For the ``branch'' $\Omega_2$, the largest
width is $b$ so that the threshold $\mu = \pi^2/b^2$.  Our goal is
therefore to find a sufficient geometrical condition to ensure that
the first eigenvalue $\lambda_1$ is smaller than $\mu$.  This
condition can be replaced by a weaker condition $\gamma_1 < \pi^2/b^2$
for the first eigenvalue $\gamma_1$ of the Laplace operator in the
trapeze $\Omega_1$ with Dirichlet boundary condition.

The eigenvalue $\gamma_1$ can be found from the Rayleigh's principle
as
\begin{equation*}
\gamma_1 = \inf\limits_{v\in H_0^1(\Omega_1)} \gamma(v) ,
\qquad \gamma(v) \equiv \frac{(\nabla v,\nabla v)_{L^2(\Omega_1)}}{(v,v)_{L^2(\Omega_1)}}. 
\end{equation*}
Taking a trial function
\begin{equation*}
v(x,y) = y \biggl(y - d + \frac{d-b}{a} x \biggr) \sin(\pi x/a) ,
\end{equation*}
which satisfies Dirichlet boundary condition on the boundary of
$\Omega_1$, we look for such conditions that $\gamma(v) < \pi^2/b^2$,
i.e.,
\begin{equation*}
Q(\ee) \equiv \frac{\pi^2}{b^2} (v,v)_{L^2(\Omega_1)} - (\nabla v,\nabla v)_{L^2(\Omega_1)}  > 0.
\end{equation*}
The direct integration yields
\begin{equation*}
Q(\ee) = \frac{b^5}{720\pi^2 a} \biggl(\kappa^2 P_A(\ee) - P_B(\ee) \biggr) ,
\end{equation*}
where $\kappa = a/b$, and $P_A(\ee)$, $P_B(\ee)$ are two polynomials
of the fifth degree:
\begin{equation}
\label{eq:PA_PB}
P_A(\ee) \equiv \sum\limits_{j=0}^5 A_j \ee^j, \qquad    P_B(\ee) \equiv \sum\limits_{j=0}^5 B_j \ee^j,
\end{equation}
with the explicit coefficients:
\begin{equation*}
\begin{array}{l l}
A_5 = 2\pi^4 - 15\pi^2 + 45\approx 91.7741      & B_5 = 2\pi^4 + 15\pi^2 - 45\approx 297.8622 \\
A_4 = 6(2\pi^4 - 10\pi^2 + 15) \approx 666.7328 & B_4 = 6(2\pi^4 + 10\pi^2 - 15)\approx 1671.0854 \\
A_3 = 30(\pi^4 - 4\pi^2 + 3)\approx 1827.9202   & B_3 = 30(\pi^4 + 3\pi^2)\approx 3810.5371 \\
A_2 = 20(2\pi^4 - 9\pi^2 + 9)\approx 2299.8348  & B_2 = 20(2\pi^4 + 3\pi^2)\approx 4488.5399 \\
A_1 = 30(\pi^4- 6 \pi^2)\approx 1145.7439       & B_1 = 30\pi^4\approx 2922.2727 \\
A_0 = 12(\pi^4 - 10\pi^2)\approx -15.4434      & B_0 = 12\pi^4\approx 1168.9091 \\
\end{array}
\end{equation*}
Note that all $B_j > 0$ and $A_j > 0$ except for $A_0 < 0$.  From the
fact that $A_j < B_j$, one has $Q(\ee) < 0$ for all $\ee>0$ when
$\kappa = 1$ (i.e., $a=b$).  We have therefore two parameters, $\ee$
and $\kappa$, which determine the sign of $Q$ and thus the exponential
decay.  Since $P_B(\ee) > 0$ for all $\ee \geq 0$, the condition
$Q(\ee) > 0$ is equivalent to two inequalities:
\begin{equation}
P_A(\ee) > 0 , \quad \kappa^2 > \frac{P_B(\ee)}{P_A(\ee)}.
\end{equation}
One can check that $P_A(\ee_0) = 0$ at $\ee_0 \approx 0.0131$ and
$P_A(\ee) > 0$ if and only if $\ee > \ee_0$ that completes the proof.
\end{proof}

We remind that this condition is not necessary (as we deal with an
estimate for the first eigenvalue).  For given $a$ and $b$ (i.e.,
$\kappa$), the above inequalities determine the values of $\ee$ (and
thus the leg $d$) for which localization occurs.  Alternatively, one
can express $a$ and $b$ through the legs $c$ and $d$ (and parameter
$\ee$) as
\begin{equation*}
a = \frac{c \ee}{\ee+1}, \qquad  b = \frac{d}{\ee+1} ,
\end{equation*}
from which $\kappa = c \ee/d$.  For given $c$ and $d$, one can vary
$\ee$ to get a family of inclosed rectangles (of sides $a$ and $b$).
The above inequalities can be reformulated as
\begin{equation*}
\begin{split}
\ee > \ee_0 \quad & \Leftrightarrow \quad b < \frac{d}{\ee_0+1}, \\
k^2 > \frac{P_B(\ee)}{P_A(\ee)}  \quad & \Leftrightarrow \quad  \frac{c}{d} > \frac{\sqrt{P_B(\ee)}}{\sqrt{P_A(\ee)}~\ee} \equiv f(\ee) . \\
\end{split}
\end{equation*}
The function $f(\ee)$ can be checked to be monotonously decreasing so
that the last inequality yields
\begin{equation}
\ee > f^{-1}(c/d)   \quad \Leftrightarrow \quad b < \frac{d}{f^{-1}(c/d)+1} ,
\end{equation}
where $f^{-1}$ denotes the inverse of the function $f(\ee)$.  This
condition determines the choice of the inscribed rectangle (the size
$b$) for a given triangle.

For the ``worst'' case $c = d$, for which a numerical computation
yields $f^{-1}(1) \approx 1.515$, one gets
\begin{equation*}
\frac{b}{d} < \frac{1}{2.515} \approx 0.3976 .
\end{equation*}
This example shows that one can always inscribe a rectangle in such a
way that $\lambda < \pi^2/b^2$.  However, the ``branch'' $\Omega_2$ in
which an exponential decay of the eigenfunction is expected, may be
small.  Figure \ref{fig:localization_triangles} illustrates these
results.

\begin{figure}
\begin{center}
\includegraphics[width=80mm]{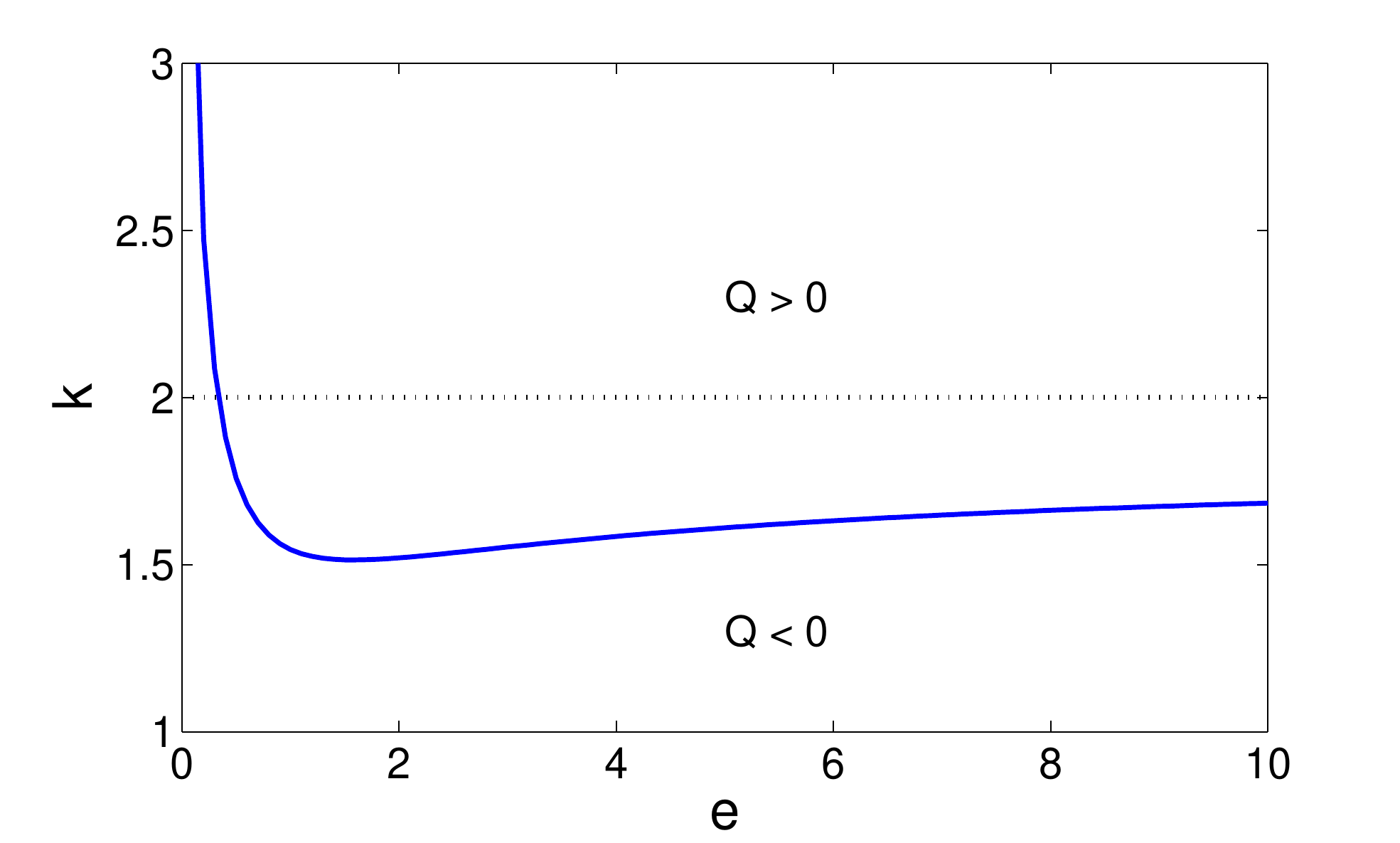}
\end{center}
\caption{
The diagram in the space of parameters $\ee$ and $\kappa$ for positive
and negative signs of $Q$ which correspond to localization and
non-localization regions.  For a given $\kappa$ (e.g., $\kappa = 2$
shown by horizontal dotted line), one can determine the values of
$\ee$, for which localization occurs.}
\label{fig:diagram}
\end{figure}

\begin{figure}
\centering
\includegraphics[width=150mm]{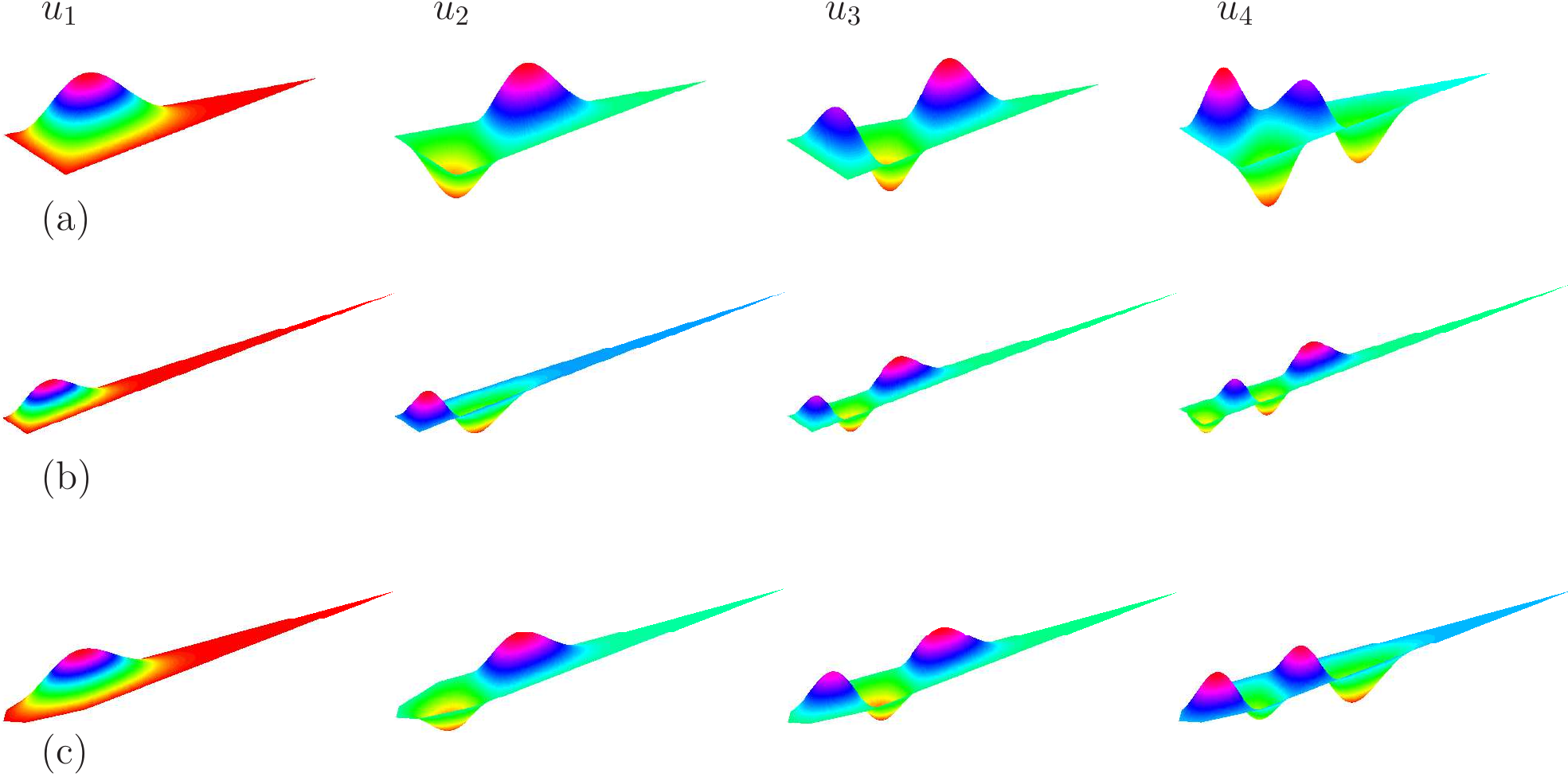}
\caption{ 
Several Dirichlet Laplacian eigenfunctions in {\bf (a)} the right
triangle with $a=2$, $b=1$ and $\ee\approx 0.32$ for which $c=8.25$
and $d=1.32$; {\bf (b)} the right triangle with $a=4$, $b=1$ and
$\ee\approx 0.08$ for which $c=61.14$ and $d=1.07$; and {\bf (c)}
elongated hexagon.  In all these cases, the first eigenvalue
$\lambda_1$ is smaller than $\pi^2$, while the associated
eigenfunction decays exponentially along the ``branch'' $\Omega_2$.
The other eigenfunctions are also concentrated in $\Omega_1$.}
\label{fig:localization_triangles}
\end{figure}

\begin{remark}
Any enlargement of the subdomain $\Omega_1$ on Fig. \ref{fig:triangle}
further diminishes the eigenvalue $\gamma_1$ and thus favors the
exponential decay in $\Omega_2$.  In particular, for each positive
integer $n$ ($n\ge 3$), one can construct elongated polygons of $n$
vertices for which the first Dirichlet Laplacian eigenfunction is
localized in $\Omega_1$ (Fig. \ref{fig:triangle}b).  Figure
\ref{fig:localization_triangles}c shows first eigenfunctions in
elongated hexagons.
\end{remark}

\end{document}